\documentclass[12pt,onecolumn,draftclsnofoot,journal]{IEEEtran}
\usepackage{graphicx}
\usepackage{float}
\usepackage{epstopdf}
\usepackage[cmex10]{amsmath}
\usepackage{array}
\usepackage{cite}
\usepackage{amssymb}
\usepackage{amsfonts}
\usepackage{amsmath}
\usepackage{stackrel}
\usepackage{booktabs,multirow}
\usepackage{arydshln}
\usepackage{amsthm}
\usepackage{listings}
\usepackage{algorithm}
\usepackage{algorithmicx}
\usepackage{algpseudocode}
\usepackage{threeparttable}
\newcommand\numberthis{\addtocounter{equation}{1}\tag{\theequation}}

\newtheorem{lem}{Lemma}

\newtheorem{theo}{Theorem}

\newtheorem{cor}{Corollary}

\newtheorem{ex}{Example}
\newfloat{routine}{htbp}{loa}
\floatname{routine}{Routine}

\makeatletter
\newcommand{\algmargin}{\the\ALG@thistlm}
\makeatother
\newlength{\forwidth}
\algdef{SE}[parFOR]{parFor}{EndparFor}[1]
  {\parbox[t]{\dimexpr\linewidth-\algmargin}{
     \hangindent\forwidth\strut\algorithmicfor\ #1\ \algorithmicdo\strut}}{\algorithmicend\ \algorithmicfor}
\algnewcommand{\parState}[1]{\State
  \parbox[t]{\dimexpr\linewidth-\algmargin}{\strut #1\strut}}

\newlength{\ifwidth}
\algdef{SE}[parIF]{parIf}{EndparIf}[1]
  {\parbox[t]{\dimexpr\linewidth-\algmargin}{
     \hangindent\ifwidth\strut\algorithmicif\ #1\ \algorithmicdo\strut}}{\algorithmicend\ \algorithmicif}

\begin{document}

\title{On Computing the Multiplicity of Cycles in Bipartite Graphs Using the Degree Distribution and the Spectrum of the Graph \thanks{This work is accepted for presentation in part at ISTC 2018, Hong Kong.}}
\author{Ali Dehghan, and Amir H. Banihashemi,\IEEEmembership{ Senior Member, IEEE}}

\maketitle


\begin{abstract}
Counting short cycles in bipartite graphs is a fundamental problem of interest in the analysis and design of low-density parity-check (LDPC) codes. The vast majority of research in this area is focused on algorithmic techniques. Most recently, Blake and Lin proposed a computational technique to count the number of cycles of length $g$ in a bi-regular bipartite graph, where $g$ is the girth of the graph. The information required for the computation is the node degree and the multiplicity of the nodes on both sides of the partition, as well as the eigenvalues of the adjacency matrix of  the graph (graph spectrum). In this paper, the result of Blake and Lin is extended to compute 
the number of cycles of length $g+2, \ldots, 2g-2$, for bi-regular bipartite graphs, as well as the number of $4$-cycles and $6$-cycles in irregular and half-regular bipartite graphs, with $g \geq 4$ and $g \geq 6$, respectively. 

\begin{flushleft}

\end{flushleft}

\end{abstract}

\section{introduction}
The performance of low-density parity-check (LDPC) codes under iterative message-passing algorithms is highly dependent on the structure of the code's Tanner graph, in general, and the distribution of short cycles, in particular, see, e.g.,~\cite{mao2001heuristic},~\cite{hu2005regular},~\cite{halford2006algorithm},~\cite{xiao2009error},~\cite{MR3071345}. 
Cycles play a particularly important role in the error floor performance of LDPC codes, where they are the main substructure of the trapping sets~\cite{asvadi2011lowering},~\cite{MR2991821},~\cite{MR3252383},~\cite{HB-CL},~\cite{HB-IT1},~\cite{HB-IT2}. The close relationship between the performance of graph-based coding schemes and
the cycle structure of the graph, especially the number of short cycles, has motivated a flurry of research activity on the study of cycle distribution and the counting of short cycles in bipartite graphs~
\cite{halford2006algorithm}, \cite{karimi2012counting}, \cite{karimi2013message}, \cite{dehghan2016new},  \cite{blake2017short}.

Counting cycles of a given length in a general graph is known to be NP-hard \cite{flum2004parameterized}.  The problem remains NP-hard even for bipartite graphs \cite{MR1405031}.
Karimi and Banihashemi\cite{karimi2013message} proposed an efficient message-passing algorithm to count the number of cycles of length less than $2g$, in a general graph, where $g$ is the girth of the graph. They also proposed 
a less complex algorithm for bipartite graphs with quasi-cyclic (QC) structure based on the relationship between the cycle multiplicities and the eigenvalues of the directed edge matrix of the graph~\cite{karimi2012counting}.   
Distribution of cycles in different ensembles of bipartite graphs was studied in~\cite{dehghan2016new}. 
It was shown in \cite{dehghan2016new} that for random irregular and bi-regular bipartite graphs, the multiplicities of cycles of different lengths have independent Poisson distributions with
the expected values only a function of the cycle length and the degree distribution, and independent of the block length.

The spectrum $\{\lambda_i\}$ of a graph $G$, defined as the eigenvalues of its adjacency matrix, is an important characteristic of $G$. 
It is known that $\sum_{i} \lambda_i=0$, $\sum_{i}\lambda_i^2=2|E(G)|$, where $|E(G)|$ is the number of edges of $G$, and
$\sum_{i}\lambda_i^3=6 N_3(G)$, where $N_3(G)$ is the number of $3$-cycles of $G$. The last result, however, cannot be extended to larger cycles, i.e., one cannot 
count cycles of length larger than $3$ as a function of only the spectrum of the graph. 
For instance, the complete bipartite graph $K_{1,4}$ (with one node on one side and four nodes on the other side of the bipartition), and the graph $   C _4 \cup K_1$ (the union of a $4$-cycle and a single node) 
have the same  spectrum
$\{-2,0, 0, 0,2\}$, but they clearly have different number of 4-cycles. 
Recently, Blake and Lin~\cite{blake2017short} computed the multiplicity of cycles of 
length $g$ in bi-regular bipartite graphs as a function of the spectrum of the graph plus the extra information about the number and the degree of the nodes on each side of the bipartition. 
In \cite{blake2017short}, it is stated: ``While only cycles of length equal to the girth are considered here, it was originally hoped that a more
detailed study would yield expressions for cycle length $g+2$ although this would involve more complex computations. The authors were unsuccessful in this but hope
this work might lead other researchers to consider the problem which could lead to a more analytical approach to code design than has yet been possible.''

Inspired by~\cite{blake2017short}, and in relation to the above statement, in this work, we extend the results of \cite{blake2017short} to compute the number of cycles of length $g+2, \ldots, 2g-2$, in bi-regular bipartite graphs in terms of the graph's degree distribution and its spectrum. 
Moreover, we compute the multiplicity of $4$-cycles in irregular graphs with $g \geq 4$, and $6$-cycles in half--regular 
graphs with $g \geq 6$, in terms of the degree distribution and the spectrum of the graph. 

Complementary to the above positive results are the negative results, presented in Table \ref{TXY1}, of the cases for which it is, in general, impossible to compute the multiplicity of cycles of a certain length $i$ in a bipartite graph of girth $g$ from only the spectrum and the degree distribution of the graph~\cite{arxiv}. These cases are denoted by ``IP," brief for ``impossible," in the table.

\begin{table}[ht]
\caption{A summary of our results on the possibility of counting cycles of length $i$ in bi-regular, half-regular and irregular bipartite graphs with girth $g$ using only the spectrum and the degree distribution of the graph.  (Notations ``P"  and ``IP'' are used for ``possible'' and ``impossible,'' respectively.)}
\begin{center}
\scalebox{1}{
\begin{tabular}{ |c|c||c|c|c|  }
\hline
 &              & $i=g$ & $i=g+2,g+2,\ldots, 2g-2$ &  $i=2g,2g+2,\ldots$\\ \hline \hline

\multirow{1}{*}{Bi-regular}
 & $g\geq 4$    & P    & P     & IP  \\ \hline \hline

\multirow{3}{*}{Half-regular}
 & $g \leq 6$        & P     & IP    & IP  \\
 & $g\geq 8$    &IP     & IP    & IP  \\ \hline \hline

\multirow{3}{*}{Irregular}
 & $g=4$        & P     & IP    & IP  \\
 & $g\geq 6$    &IP     & IP    & IP  \\ \hline

\end{tabular}
}
\end{center}
\label{TXY1}
\end{table}

The organization of the rest of the paper is as follows: In Section~\ref{sec1}, we present some definitions and notations. This is followed in Section~\ref{sec2} by our results
on computing the number of cycles of length $g+2, \ldots, 2g-2$, in bi-regular bipartite graphs using the spectrum and the degree distribution of the graph.
In Section~\ref{sec35}, we first consider irregular bipartite graphs with $g \geq 4$, and compute the number of $4$-cycles in such graphs as a function of the graph spectrum and its degree distribution. We then derive a similar result for counting $6$-cycles in half-regular bipartite graphs with $g \geq 6$. 
Section~\ref{sec4} is devoted to numerical results. The paper is concluded with some remarks in Section~\ref{sec5}.

\section{Definitions and notations}
\label{sec1}

For a given graph $G$, we denote the node set and the edge set of $G$ by $V(G)$ and $E(G)$, 
respectively. The shorthands $V$ and $E$ are used if there is no ambiguity about the graph. In this work, we consider undirected graphs with no loops or parallel edges (i.e., simple undirected graphs). 
An edge $e \in E$ with endpoints $u \in V$ and $w \in V$ is denoted by $\{u,w\}$, or by $uw$ or $wu$, in brief. 
The number of edges incident to a node $v$ is called the {\em degree} of $v$, and is denoted by $d(v)$.
For a given graph $G$, a {\it walk} of length $k$ is a sequence of nodes
$v_1, v_2, \ldots , v_{k+1}$ in $V$ such that $\{v_i, v_{i+1}\} \in E$, for all $i \in \{1, \ldots , k\}$. Equivalently, a walk of length $k$ can be described
by the corresponding sequence of $k$ edges.
Let $\mathcal{W}=v_1, v_2, \ldots , v_{k+1}$ be a walk in the graph $G$, we say that $\mathcal{W}'=v_{l_0},v_{l_1}, \ldots, v_{l_s}$, for $s \geq 1$, is a subwalk of $\mathcal{W}$ if there is an index $i$, $1  \leq i \leq k-s+1$, such that $v_i=v_{l_0}, v_{i+1}=v_{l_1}, \ldots, v_{i+s}=v_{l_s}$.
A walk $v_1, v_2, \ldots , v_{k+1}$ is a {\it path} if all the nodes $v_1, v_2, \ldots , v_k$ are distinct. A walk is called a
{\it closed walk}  if the two end nodes are the same, i.e.,
if $v_1 = v_{k+1}$. Under the same condition, a path is called a {\it cycle}. In the rest of the paper, the term ``path'' is used only to refer to the paths that are not cycles. We also use the notation $P_n$ to denote a path with $n$ nodes.
The {\em length} of a walk, path or cycle is the number of its edges.
We denote cycles of length $k$, also referred to as $k$-cycles, by ${\cal C}_k$. We use $N_k$ for $|{\cal C}_k|$.  The length of the shortest cycle(s) in a graph is called {\em girth} and is denoted by $g$.

A graph $G$ is {\it connected}, if there is a path between any two nodes of $G$. If the graph $G$ is not connected, we say that it is disconnected. A {\it connected  component}  of  a  graph  is  a  connected
subgraph such that there are no edges between nodes of the subgraph and nodes of the rest of the graph.

A graph $G=(V,E)$ is called {\it bipartite}, if the node set $V$ can be
partitioned into two disjoint subsets $U$ and $W$, i.e., $V = U \cup W \text{ and } U \cap W =\emptyset $, such that every edge in $E$ connects a node
from $U$ to a node from $W$. A graph is bipartite if and only if the lengths of all its cycles are even.
Tanner graphs of LDPC codes are bipartite graphs, in which $U$ and $W$ are referred to as {\it variable nodes} and {\it check nodes}, respectively. 
Parameters $n$ and $m$ in this case are used to denote $|U|$ and $|W|$, respectively. Parameter $n$ is the code's block length and the code rate 
$R$ satisfies $R \geq 1- (m/n)$. 

The {\it degree sequences} of a bipartite graph $G$ are defined as the two monotonic non-increasing sequences of the node degrees on the two sides of the graph. 
For example, the complete bipartite graph $K_{2,3}$ has degree sequences $(3,3)$ and $(2,2,2)$. Clearly, the degree sequences also contain the 
information about the number of nodes on each side of the graph.
A bipartite graph $G = (U\cup W,E)$ is called {\it bi-regular}, if all the nodes on the same side of the bipartition have the same degree,
i.e., if all the nodes in $U$ have the same degree $d_u$ and all the nodes in $W$ have the same degree $d_w$.
It is clear that, for a bi-regular graph, $|U|d_u=|W|d_w=|E(G)|$. A bipartite graph is called {\em half-regular}, if all the nodes on one side of the bipartition have the same degree.
A half-regular Tanner graph can be either variable-regular or check-regular. A Tanner graph $G = (U\cup W,E)$ is called variable-regular with variable degree $d_v$,
if for each variable node $u_i\in U$, $d({u_i}) = d_v$. Similarly, a Tanner graph is called {\em check-regular} with check degree $d_c$,
if for each check node $w_i\in W$, $d({w_i}) = d_c$. Also, a $(d_v, d_c)$-regular Tanner graph is a bi-regular graph with variable degree $d_v$ and check degree $d_c$.
A bipartite graph that is not bi-regular is called {\it irregular}. With this definition, half-regular graphs are a special case of irregular graphs.

A bipartite graph $G(U \cup W, E)$ is called {\em complete}, and is denoted by $K_{|U|,|W|}$, if every node in $U$ is connected to every node in $W$. 
The notation $K_m$ is used for a complete (non-bipartite) graph with $m$ nodes (in which every node is connected to all the other nodes).

A {\it tree} is an undirected graph in which any two nodes are connected by exactly one path. Any  connected graph is a tree if and only if it does not have any cycle. 
A {\em rooted tree} is a tree in which one node is designated as the root. 
In a given tree, a node $v$ is called {\it leaf} if $d(v) = 1$. The {\it height} of a node in a rooted 
tree is the length of the longest path from that node to a leaf when moving away from the root. The height of the tree is the height of the root. 

Consider the graph $G = (V, E)$, and let $S \subset V$ be any subset of nodes of $G$. Then, the {\it node-induced subgraph} (or simply ``induced subgraph") on the set of nodes $S$ is the graph whose node set is $S$ and whose edge set consists of all the edges in $E$ that have both endpoints in $S$. Similarly, an {\it edge-induced subgraph}  on the set of edges $D \subset E$ is the graph that consists of the edges $D$ together with any nodes that are the endpoints of the edges in $D$.

The {\it adjacency matrix} of a graph $G$ is the matrix $A = [a_{ij}]$, where $a_{ij}$ is the number of edges connecting the node $i$ to the node
$j$ for all $i, j\in V$. The matrix $A$ is symmetric and since we have assumed that $G$ has no parallel edges or loops, $a_{ij}\in\{0, 1\}$,
for all $i, j\in V$, and $a_{ii} = 0$, for all $i \in V$. The set of the eigenvalues $\{\lambda_i\}$ of $A$ is called the {\em spectrum} of the graph. 
It is well-known that the spectrum of a disconnected graph is  the disjoint union of the spectra of its components \cite{MR2882891}.
One important property of the adjacency matrix is that the number of walks
between any two nodes of the graph can be determined using the powers of this matrix. More precisely, the entry in
the $i^{\text{th}}$ row and the $j^{\text{th}}$ column of $A^k$, $[A^k]_{ij}$ , is the number of walks of length $k$ between nodes $i$ and $j$. In particular, $[A^k]_{ii}$
is the number of closed walks of length $k$ containing node $i$. The total number of closed walks of length $k$ in $G$ is thus $tr(A^k)$, where $tr(\cdot)$ is the trace of a matrix.\footnote{In this calculation, closed walks with the same set of edges but with different starting edge or with different direction are distinguished and counted separately.}
Since $tr(A^k)= \sum_{i=1}^{|V|}\lambda_i^k$, it follows that the multiplicity of closed walks of different length in a graph can be obtained using the spectrum of the graph.

In the figures of this paper, edges of a graph are shown by straight lines. In this work, we need to closely examine different types of closed walks. 
To show a closed walk in a graph, we draw a closed curved line alongside the graph following the edges of the walk. 
We place an arrow at the location corresponding to the starting edge of the walk to specify the starting edge and  the direction of the walk. An example is shown in Fig. \ref{graphBBB}.

\begin{figure}[ht]
\begin{center}
\includegraphics[scale=.40]{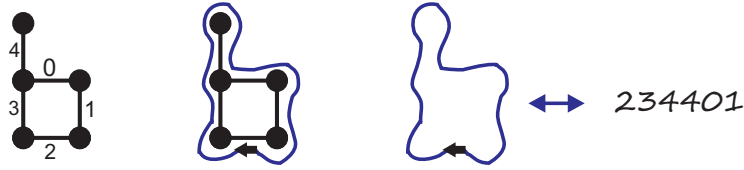}
\caption{The representation  of a graph and one of its closed walks.
} \label{graphBBB}
\end{center}
\end{figure}


In general, the spectrum of a graph does not uniquely determine the graph. Two graphs are called {\it cospectral} or {\it isospectral} if they have the same spectrum.
For example, the complete bipartite graph $K_{1,4}$, and the graph $ C _4 \cup K_1$ (the union of a 4-cycle and a single node) are cospectral, both having the spectrum $\{-2,0,0,0,2\}$
or $\{-2, 0^3, 2\}$. 
On the other hand, there are graphs that are known to be uniquely determined by their spectrum.
Two examples are the complete graph $K_n$, and the complete bipartite graph $K_{n,n}$ \cite{MR2022290}.

It is well-known that a number of properties of a graph $G$ can be uniquely specified based on the information of the graph's spectrum (see, for example, \cite{MR2022290}). 
These properties include the number of nodes and edges of $G$, the number of cycles of length three in $G$, as well as properties that involve a binary question such as whether $G$ is regular or not, whether $G$ is regular with any fixed girth or not, 
and whether $G$ is bipartite or not. 
In particular, a given graph is bipartite if and only if its spectrum is symmetric with respect to the origin.
On the other hand, there are some other important properties of a graph, such as the number of cycles of length larger than three, that cannot be 
determined by the spectrum alone. In this work, we are interested in counting the number of short cycles in bipartite graphs. In particular, to complement the results of \cite{blake2017short}, we investigate 
whether such counting problems can be solved for cycles larger than the girth  in bi-regular bipartite graphs, or for cycles in bipartite graphs that are not bi-regular, by using the spectrum of the graph and 
the extra information about the node degrees of the graph. 

\section{Computing the multiplicity of short cycles in bi-regular bipartite graphs}
\label{sec2}

In this section, we compute the multiplicity of $k$-cycles of bi-regular bipartite graphs with $g \geq 4$, for $g+2 \leq k \leq 2g-2$, in terms of the spectrum and the degree sequences of the graph.
The results presented in this section complement those of Blake and Lin~\cite{blake2017short} for $g$-cycles.
The results are obtained by characterizing and counting closed walks of length $k$ that are not cycles, and subtracting their multiplicity from the total number of 
closed walks of length $k$. The latter is easily obtained using the spectrum of the graph. In this section, we also provide a brief review of the main result of \cite{blake2017short}, and 
propose an alternate approach for the calculation of the number of closed cycle-free walks in a bipartite graph.

\subsection{Categorization of closed walks}

A closed walk $\mathcal{W}$ is called {\em cycle-free} if the edge-induced subgraph on the set of edges of $\mathcal{W}$ does not have any cycle.
An example of a closed cycle-free walk  is shown in Fig.~\ref{V2graphB1}(a). We say a closed walk $\mathcal{W}$ is a {\it closed walk with cycle}, or CWWC, in brief, if $\mathcal{W}$ is not a cycle but the edge-induced subgraph on the set of edges of  $\mathcal{W}$ has at least one cycle. An example of a closed walk with cycle is shown in Fig. \ref{V2graphB1}(b). This closed walk has length $10$ and traverses through the edge $uu'$ three times. We note that if $uu'$ is traversed only once, we still have a CWWC but of length $8$. 

\begin{figure}[ht]
\begin{center}
\includegraphics[scale=.4]{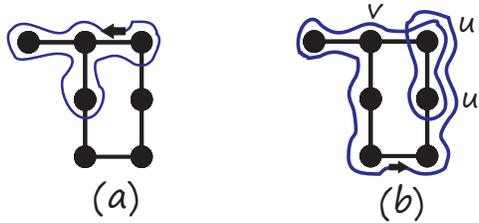}
\caption{(a) An example of  a closed cycle-free walk of length six, (b) An example of a closed walk of length 10 with cycle or a $10$-CWWC.
} \label{V2graphB1}
\end{center}
\end{figure}

\begin{lem}\label{V2L1}
All the closed walks of length $k$ in a graph $G$ can be partitioned into three categories: (i) $k$-cycles, (ii) closed cycle-free walks of length $k$, and (iii) closed walks of length $k$ with cycle.
\end{lem}

In this work, two walks $e_{i_1},\ldots, e_{i_r} $ and $e_{j_1},\ldots, e_{j_r}$ are considered identical, and thus counted as one, if and only if for every $x$ in the range $1\leq x \leq r$, $e_{i_x}=e_{j_x}$.
In other words, closed walks that pass through the same set of edges but in different directions or with different starting edge are considered distinguishable and counted separately. 
The following theorem then follows immediately from Lemma~\ref{V2L1}.

\begin{theo}\label{L1}
For a  given $(d_v, d_c)$-regular bipartite graph $G$, the number of $i$-cycles is given by:
\begin{equation}\label{E3}
N_i = [\sum_{j=1}^{|V|} \lambda_j^i - \Omega_{i}(d_v,d_c,G) - \Psi_{i}(d_v,d_c,G)]/(2i),
\end{equation}
where $\{\lambda_j\}_{j=1}^{|V|}$ is the spectrum of $G$, and $\Omega_{i}(d_v,d_c,G)$ and $\Psi_{i}(d_v,d_c,G)$ 
are the number of closed cycle-free walks of length $i$ and closed walks with cycle of length $i$ in $G$, respectively.
\end{theo}

The multiplicity $\Omega_{i}(d_v,d_c,G)$ of closed cycle-free walks of length $i$ in a $(d_v, d_c)$-regular bipartite graph $G$ was computed in \cite{blake2017short}. In the following, we review the result of \cite{blake2017short} and also provide an alternate approach for the computation.

\subsection{Calculation of $\Omega_{i}(d_v,d_c,G)$}

\subsubsection{Approach of \cite{blake2017short}}
Blake and Lin \cite{blake2017short} used the following formula to calculate $\Omega_{i}(d_v,d_c,G)$, for $2 \leq i \leq 2g-2$:
\begin{equation}\label{E2}
\Omega_{i}(d_v,d_c,G) = n\times S_{d_v,d_c,i}+m\times S_{d_c,d_v,i}\:.
\end{equation}
In (\ref{E2}), parameters $n$ and $m$ are the number of variable and check nodes in $G$, respectively, and $S_{d_v,d_c,i}$ ($S_{d_c,d_v,i}$) represents the  number of  
closed cycle-free walks of length $i$ from a variable node $v$ (a check node $c$) to itself. Generating functions were then used to compute the functions $S_{x,y,i}$ recursively.
In Table \ref{T1}, we have shown the functions $S_{x,y,i}$ for values of $i$ up to ten.

\begin{table}[ht]
\scriptsize
\caption{In a  $(d_v, d_c)$-regular bipartite graph, the number of  closed cycle-free walks  of length $i$ from any variable node $v$ (check node $c$) to itself is equal to $S_{d_v,d_c,i}$ ($S_{d_c,d_v,i}$) \cite{blake2017short}}
\begin{center}
\scalebox{1}{
\begin{tabular}{|r|c|c|}
\hline
$i$ & $Q_{x,y,i}$ & $S_{x,y,i}$ \\ \hline
$2$ & $x$ & $x$ \\ \hline
$4$ & $x(y-1)$  & $x(x+y-1)$ \\ \hline
$6$ & $x((y-1)^2 +(x-1)(y-1))$ & $x(x^2 + 2x(y-1) + (x-1)(y-1) + (y-1)^2 )$ \\ \hline
$8$ & $x((y-1)^3 + 3(x-1)(y-1)^2 )$ & $x\Big( (y-1)^3 + 3(x-1)(y-1)^2 + (x-1)^2 (y-1)  \Big)$ \\
    &   $+x (x-1)^2 (y-1)$     &  $+ x\Big( 2x((y-1)^2 + (x-1)(y-1))+x(y-1)^2 +3x^2 (y-1) + x^3 \Big)$ \\ \hline
$10$ & $x((y-1)^4 + 6(x-1)(y-1)^3 )$ & $Q_{x-1,y-1,10} + 2Q_{x-1,y-1,2} Q_{x-1,y-1,6} + 2Q_{x-1,y-1,4} Q_{x-1,y-1,6} $\\
   &  $+x(6(x-1)^2 (y-1)^2 + (x-1)^3 (y-1))$  & $+ 3Q_{x-1,y-1,2} Q_{x-1,y-1,4}^2 + 3Q_{x-1,y-1,2}^2 Q_{x-1,y-1,6} + $ \\
   & & $4Q_{x-1,y-1,2}^3 Q_{x-1,y-1,4} + Q_{x-1,y-1,2}^5 $\\
\hline
\end{tabular}
}
\end{center}
\label{T1}
\end{table}

\subsubsection{Alternate approach}

Let $T_{d_v,d_c,i}$ be a rooted tree of height $\frac{i}{2} $ with the root node of degree $d_v$ at level zero, and the nodes of succeeding levels with alternating degrees $d_c$ and $d_v$, in odd and even levels of the tree, respectively.  
In such a tree, all the leaves are in level $\frac{i }{2}$. As an example, the tree $T_{3,4,4}$ is shown in Fig. \ref{V4g1}.
Let $A(T_{d_v,d_c,i})$ be the adjacency matrix of $T_{d_v,d_c,i}$ such that the root corresponds to the first row of the matrix. Then for a given $(d_v, d_c)$-regular Tanner graph $G$,  
the number of  closed cycle-free walks  of length $i$ from a variable node $v$ in $G$ to itself (i.e., $S_{d_v,d_c,i}$) is equal to the $(1,1)$-th entry of the matrix $A(T_{d_v,d_c,i})^i$. This follows from the fact that there are no cycles in $T_{d_v,d_c,i}$, and thus all the closed walks are cycle-free.
Similarly, $S_{d_c,d_v,i}$ is equal to the $(1,1)$-th entry of the matrix $A(T_{d_c,d_v,i})^i$. Therefore, to obtain $S_{x,y,i}$ for different values of $x$, $y$ and $i$, one can form the adjacency matrix of 
$T_{x,y,i}$, calculate its $i$-th power and then take the $(1,1)$-th entry of the resulting matrix. 
Using this technique, we have calculated $S_{x,y,i}$ for some practical values of $x$, $y$, and $i=10, 12$. These are provided in Table \ref{T2}. 
\begin{figure}[ht]
\begin{center}
\includegraphics[scale=.35]{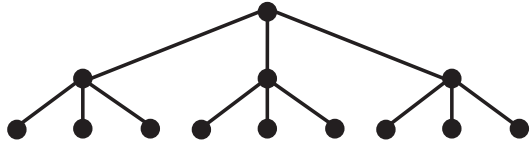}
\caption{The tree $T_{3,4,4}$.
} \label{V4g1}
\end{center}
\end{figure}

\begin{table}[ht]
\scriptsize
\caption{The number of  closed cycle-free walks of length $i$ in an $(x,y)$-regular bipartite graph from a node with degree $x$ to itself for different $x$, $y$ and $i$ (i.e., $S_{x,y,i}$). }
\begin{center}
\scalebox{1}{
\begin{tabular}{ |c||c|c|c|c|c|c|c|  }
\hline
       & $x=2$    & $x=3$    & $x=4$  & $x=5$  & $x=6$ & $x=7$ & $x=8$ \\
\hline
\multicolumn{8}{|l|}{$i=10$} \\
\hline
$y=2$ &  252    & 1278       & 4144    & 10500  & 22716 & 44002 & 78528 \\
\hline
$y=3$ &  852     & 3543      & 10104   & 23325  & 46956 & 85827 & 145968\\
\hline
$y=4$ &  2072   &  7578      & 19864  & 43100  & 82656 & 145222 & 238928\\
\hline
$y=5$ &  4200   &  13995     & 34480  & 71445  & 132120 &225295 &361440\\
\hline
$y=6$ &  7572   &  23478     &  55104 & 110100 & 197796& --    &  --\\
\hline
$y=7$ &  12572  &  36783     & 82984  & 160925 &282276 & --    &  -- \\
\hline
$y=8$ &  19632  &  54738     & 119464  & 225900& --    & --   & --\\
\hline
\multicolumn{8}{|l|}{$i=12$} \\
\hline
$y=2$ &    --    & 6486      & 26408   & 79860  & 199812 &438074&  871056\\
\hline
$y=3$ & 4324    & 23823      & 82920  & 223795  & 512748 & --   &  -- \\
\hline
$y=4$ & 13204   &  62190     & 195352 & 488980   & --   & --   &  --\\
\hline
$y=5$ & 31944   &  134277    & 391184 & --       & --   & --   &  --\\
\hline
$y=6$ &  66604  &  256374    & --     & --       & --   & --   &  --\\
\hline
$y=7$ &  125164 &  448731    & --     & --       & --   & --   &  -- \\
\hline
$y=8$ & 217764  &  --        & --     & --       & --   & --   &  --\\
\hline
\end{tabular}
}
\end{center}
\label{T2}
\end{table}

\subsection{Calculation of $N_{g}$}

It is clear that $\Psi_{g}(d_v,d_c,G)=0$. Based on this and Theorem \ref{L1}, we have the following corollary.

\begin{cor} \cite{blake2017short} \label{Th1}
The total number of cycles of length $g$ in a $(d_v, d_c)$-regular bipartite graph $G(V,E)$ is equal to:
\begin{equation}
\label{E1}
N_g =
\dfrac{\sum_{i=1}^{|V|}\lambda_i^g - \Omega_{g}(d_v,d_c,G)}{2g},
\end{equation}
where $\{\lambda_i\}$ is the spectrum of $G$ and $\Omega_{g}(d_v,d_c,G)$ is given by (\ref{E2}).
\end{cor}

In the following, we consider an example of a bi-regular bipartite graph, for which the number of short cycles can be computed using simple combinatorial 
arguments. We use the same example also in Subsection \ref{sub234} to demonstrate that the results obtained by our computational technique match the results obtained by combinatorics.

\begin{ex}\label{Example1}
Consider the complete bipartite graph $K_{x,x}$. For this graph, $g=4$, and we are interested in counting the number of $4$-cycles and $6$-cycles ($2g-2 = 6$). 
Let $i$ be 4 or 6. To find an $i$-cycle, one needs to choose $i/2$ nodes out of the $x$ nodes on each side of the graph with ordering. This results in
$$ N_i= \frac{1}{i}\Big(\frac{x!}{(x-\frac{i}{2})!}\Big)^2\:,$$
where the division by $i$ is due to the fact that in the above process, each cycle is counted $i$ times.
From the above formula, we have: $N_4=x^2(x-1)^2/4$ and $N_6=x^2(x-1)^2(x-2)^2/6$.
Now, we use Corollary~\ref{Th1} to calculate the number of $4$-cycles. The eigenvalues of $K_{x,x}$ are $\{0^{2x-2}, x, -x\}$. Thus, $\sum_{i} \lambda_i^g = 2x^4$. 
From Table \ref{T1}, we have: $S_{x,x,4}=x(2x-1)$.  By replacing this in (\ref{E2}), we therefore have: $\Omega_{4}(x,x,G)=2x^2(2x-1)$. 
Consequently, by Corollary \ref{Th1}, we obtain $N_4=\Big(2x^4 - 2x^2(2x-1)\Big)/8=x^2(x-1)^2/4$, which is the same result as the one we derived by combinatorial arguments.
\end{ex}

\subsection{Properties and characterization of CWWCs (of length at most $2g-2$)}

To calculate $\Psi_{i}(d_v,d_c,G)$ for a $(d_v, d_c)$-regular bipartite graph $G$, in the following, we first study some important properties of closed walks with cycles.

\begin{lem}\label{V2L2}
Let $G$ be a bipartite graph with girth $g$. If $\mathcal{W}$ is a closed walk of length $i$ with cycle in $G$, where $i\leq 2g-2$,
then there is an edge $e$ in $G$ such that $\mathcal{W}=\mathcal{W}'ee\mathcal{W}''$, or $\mathcal{W}= ee\mathcal{W}'$, or $\mathcal{W}=\mathcal{W}'ee$, or $\mathcal{W}=e\mathcal{W}'e$, where $\mathcal{W}'$ and $\mathcal{W}''$ are subwalks of $\mathcal{W}$.
\end{lem}

\begin{proof}
We prove this by contradiction. Suppose that $\mathcal{W}$ does not have an edge with the property described in the lemma. We start from $v$, the first node of $\mathcal{W}$, and visit 
the edges of $\mathcal{W}$ one by one until we reach the first repeated node $u$. By our assumption, the closed subwalk from $u$ to itself is a cycle.  
Call that cycle $\mathcal{W}_1$, and remove it from $\mathcal{W}$. The remaining graph is another closed walk from $v$ back to itself and its size is at least four (otherwise, an edge $e$, as described in the lemma must exist). Call this closed walk $\mathcal{W}_2$. 
Using a similar argument as before, the closed walk $\mathcal{W}_2$ must contain a cycle $\mathcal{W}_3$. 
We thus conclude that $\mathcal{W}$ has at least two edge disjoint cycles $\mathcal{W}_1$ and $\mathcal{W}_3$. This implies $i \geq 2g$, which is a contradiction. 
\end{proof}

\begin{lem}\label{V2L3}
Let $G$ be a bipartite graph with girth $g$ and $\mathcal{W}$ be a closed walk of length $i$ in $G$, $i \geq g$.
If there is an edge $e=vu$ that appears only once in $\mathcal{W}$, then the edge-induced subgraph on the set of edges of 
$\mathcal{W}$ has a cycle.
\end{lem}

\begin{proof}
Let  $e=vu$ be an edge that appears only once in $\mathcal{W}$. Consider the edge-induced subgraph $G'$ on the set of edges of $\mathcal{W}\setminus \{e\}$. 
This subgraph is connected. In $G'$, consider the shortest path from $u$ to $v$ and call it $P$. 
The union of $P$ and $e$ is a cycle. 
\end{proof}

Consider the CWWC $\mathcal{W}$ shown in Fig. \ref{V2graphB1}(b). The closed walk $\mathcal{W}$ consists of a $6$-cycle $\zeta$ and two 
closed cycle-free walks of length two from nodes $u$ and $v$ (of cycle $\zeta$) to themselves.
In the following, we prove that any $i$-CWWC ($i\leq 2g-2$) consists of one cycle ${\zeta}$ and some closed cycle-free walks 
from the nodes of $\zeta$ to themselves.

\begin{lem}\label{L5}
Let $G$ be a  bipartite graph with girth $g$. If $\mathcal{W}$ is an $i$-CWWC, where $i \leq 2g-2$,
then, the walk $\mathcal{W}$ consists of one cycle $\zeta$ and some closed cycle-free walks from the nodes of $\zeta$ to themselves.\footnote{This result is used later to count the number of CWWCs irrespective of the direction and the starting edge of the walk. The result of this lemma should thus be interpreted accordingly, i.e., for $\mathcal{W}$ to be formed, all the edges of $\zeta$ are traversed in a given direction, and for each cycle-free subgraph ${\cal T}$ attached to one of the nodes of $\zeta$, say $v$, each edge of ${\cal T}$ is traversed an even number of times equally in each direction. This means that, in general, there is no requirement that all the traversals through the edges of ${\cal T}$ happen sequentially. Nor is there a requirement that the traversals are initiated from $v$.}
\end{lem}

\begin{proof}
We prove the claim by induction on the number of edges of the closed walk. \\
\underline{Basis:} The smallest CWWC $\mathcal{W}$ 
has $g+2$ edges.
By definition, the edge-induced subgraph on the set of edges of $\mathcal{W}$ has a cycle, and on the other hand, $\mathcal{W}$ is not a $(g+2)$-cycle. Thus, $\mathcal{W}$
consists of one $g$-cycle  and a closed cycle-free walk of length $2$ from a node of that $g$-cycle to itself.\\
\underline{Induction step:} Suppose that the claim is true for any closed walk of length $i$ with cycle, where $i \leq 2g-4$. Now, we prove the claim for all $(i+2)$-CWWCs.  
Let $\mathcal{W}$ be such a closed walk. By Lemma \ref{V2L2}, there is an edge $e=uv$ such that $\mathcal{W}=\mathcal{W}'ee\mathcal{W}''$, or $\mathcal{W}= ee\mathcal{W}'$, or $\mathcal{W}=\mathcal{W}'ee$, or $\mathcal{W}=e\mathcal{W}'e$. Remove the two copies of $e$, just described, from $\mathcal{W}$, and call the remaining closed walk $\mathcal{W}_1$.  Now, two cases can be considered:\\
{\bf Case 1.} If the edge-induced subgraph on the set of edges of $\mathcal{W}_1$ has a cycle, then $\mathcal{W}_1$ is a closed walk of length $i$ with cycle. Thus, by the induction hypothesis, 
we know that $\mathcal{W}_1$ consists of one cycle $\zeta$ and some closed cycle-free walks from the nodes of $\zeta$ to themselves. It is then easy to see that the same also applies to $\mathcal{W} $.\\
{\bf Case 2.} Suppose that the edge-induced subgraph on the set of edges of $\mathcal{W}_1$ does not have any cycle.
Since, the edge-induced subgraph on the set of edges of $\mathcal{W}$ has a cycle, we conclude that $e=uv$ is neither in $\mathcal{W}_1$ nor in the edge-induced subgraph on the set of edges of $\mathcal{W}_1 $. There is, however, 
a path between the nodes $v$ and $u$ in the edge-induced subgraph of $\mathcal{W}_1 $. Call this path $\mathcal{P}$. The length of $\mathcal{P}$ is at least $g-1$.
So, $\mathcal{W}_1$ has at least $g-1$ different edges. On the other hand, by Lemma \ref{V2L3}, each edge of $\mathcal{W}_1$ appears at least twice in $\mathcal{W}_1$. Thus, the length  of $\mathcal{W}_1$ is at least $2g-2$, which implies that the length of $\mathcal{W}$ is at least $2g$. But this is a contradiction.  So, this case does not occur.
\end{proof}

If  $\mathcal{W}$ is a CWWC of length $i$ and $i\leq 2g-2$, then it is clear that the edge-induced subgraph on the set of edges of $\mathcal{W}$ does not have two edge-disjoint cycles. 
In the next lemma, we show that the subgraph has exactly one cycle.

\begin{lem}\label{V2L4}
Let $G$ be a bipartite graph with girth $g$. If $\mathcal{W}$ is a CWWC of length $i$ and $i\leq 2g-2$,
then, the edge-induced subgraph on the set of edges of $\mathcal{W}$ has exactly one cycle.
\end{lem}

\begin{proof}
By Lemma \ref{L5}, we know that $\mathcal{W}$ consists of one cycle $\zeta$ and some closed cycle-free walks from the nodes of $\zeta$.
Now, we prove that the cycle $\zeta$ is the only cycle in  the edge-induced subgraph on the set of edges of $\mathcal{W}$. To the contrary, assume that the edge-induced subgraph on the set of edges of $\mathcal{W}$
has another cycle $\zeta'$ such that $\zeta$ and $\zeta'$ share 
$\ell \geq 1$ edges. Denote the length of cycle $\zeta$ ($\zeta'$) by $L( \zeta )$ ($L( \zeta')$). 
Since the union of two cycles minus their shared edges contains at least one cycle, and since the girth of the
graph $G$ is $g$, we have:
\begin{equation}\label{V2E1}
L( \zeta )+ L( \zeta ')-2\ell \geq g
\end{equation}
On the other hand, by  
Lemma \ref{L5}, the closed walk $\mathcal{W}$ visits every edge of $\zeta'$ which is not in $\zeta$ at least twice. We therefore have:
\begin{equation}\label{V2E2}
i\geq L( \zeta )+ 2(L( \zeta ')-\ell )
\end{equation}
By combining (\ref{V2E1}), (\ref{V2E2}), and the fact that $L( \zeta')\geq g$, we obtain $i \geq 2g$. But this is a contradiction.
\end{proof}

We note that the closed cycle-free walks that start from the nodes of cycle $\zeta$, as described in Lemma~\ref{L5}, can have some edges in common with $\zeta$. For instance, in the CWWC of length $10$ shown in Fig. \ref{V2graphB1}(b), the closed cycle-free walk of length $2$ from node $u$ traverses twice through one of the edges of the $6$-cycle. The following result, whose proof is simple, shows that at least one of the edges of ${\zeta}$ appears only once in the closed walk with cycle.

\begin{lem}\label{L6}
If $G$ is a  bipartite graph with girth $g$, then for each $i$,  $g+2 \leq i \leq 2g-2$,
every closed walk $\mathcal{W}$ of length $i$ with cycle has at least one edge that appears only once in $\mathcal{W}$.
\end{lem}

The next result follows from Lemma~\ref{L6} by choosing the edge $e$ that appears exactly once in $\mathcal{W}$ as the $j$-th edge of $\mathcal{W}$, for any $j$ in the range $1 \leq j \leq g+k$, combined with 
the two directions that can be selected for traversing the edges of $\mathcal{W}$.

\begin{lem}\label{V2L5}
Consider a bipartite graph $G$ with girth $g$, and let $k< g$. We can then divide $\Psi_{g+k}(d_v,d_c,G)$ by $2(g+k)$ to obtain the number of 
CWWCs of length $(g+k)$ irrespective of the direction and the starting edge of the closed walk.
\end{lem}

The general approach that we use to calculate $\Psi_{g+k}(d_v,d_c,G)$ is based on employing Lemmas \ref{L5} and \ref{V2L4} to count different types 
of CWWCs of a certain length irrespective of the direction or the starting edge of the walk. We then use Lemma~\ref{V2L5} to account for the direction and the starting edge.
Suppose that we are interested in counting the CWWCs of length $i$ that consist of a cycle of length $g+k$ and some closed cycle-free walks from the nodes of that cycle, where
$i \leq 2g-2$, and $k< g-2$.
By Lemmas \ref{L5} and \ref{V2L4}, in this case, each CWWC
consists of a $(g+k)$-cycle and some closed cycle-free walks $\mathcal{W}_1,\ldots, \mathcal{W}_j $ from the nodes of the cycle, such that $g+k + \sum_{l=1}^{j} |\mathcal{W}_l|= i$. 
Considering that the length of each closed walk is even, we thus have  $1 \leq j \leq \frac{i-g-k}{2}$.  
To count the CWWCs under consideration, we need to partition the number $\frac{i-g-k}{2}$ into $j$ positive integer numbers (with $j$ in the above range), where each integer number represents half of the length of one of the closed cycle-free walks. 
The number of ways this partitioning can be performed determines the number of possibilities for the lengths of $j$ closed cycle-free walks. In the following, corresponding to each partitioning, we 
identify a {\em category} of CWWCs, i.e., in a given category, the lengths of closed cycle-free walks are fixed. Within each category, we then identify all the possibilities that closed cycle-free walks with the given lengths can be attached to the nodes of the $(g+k)$-cycle. 

An {\it integer partition} of a positive integer $n$ is defined as a way of describing $n$ as a sum of positive integers.
For example, the integer number $4$ has five integer partitions: $4$, $3 + 1$, $ 2 + 2$, $2 + 1 + 1$, and $1 + 1 + 1 + 1$. 
The number of integer partitions of $n$ is given by the {\em partition function} $p(n)$. For the example just given, $p(4) = 5$. 
An asymptotic expression for $p(n)$ is given by \cite{MR1634067}
\begin{equation}\label{E12}
p(n) \sim \frac {1} {4n\sqrt3} \exp\Big({\pi \sqrt {\frac{2n}{3}}}\Big)\:.
\end{equation}
In this work, however, we are interested in relatively short closed walks with cycles, where the length of the walk is at most $2g-2$.

\subsection{Calculation of $N_{g+2}$}
\label{sub234}
To use Theorem \ref{L1} for the calculation of $N_{g+2}$, we need to calculate $\Psi_{g+2}(d_v,d_c,G)$.

\begin{theo}\label{L2}
For a $(d_v, d_c)$-regular bipartite graph $G$ with girth $g$, we have 
$$\Psi_{g+2}(d_v,d_c,G)=N_g\times \Big(\frac{g}{2}(d_v+d_c)-g\Big)\times 2(g+2)\:.$$
\end{theo}

\begin{proof}
We first focus on counting CWWCs of length $g+2$ irrespective of their starting edge or direction. By Lemma \ref{V2L5}, we then need to multiply the obtained value by $2(g+2)$ to take into account the different starting edges and directions.
By Lemma \ref{L5},  every $(g+2)$-CWWC consists of a cycle of length $g$ and some closed cycle-free walks. Since for $i=g+2$ and $k=0$, we have $\frac{i-g-k}{2}=1$, and since $p(1)=1$,
there is only one possibility for the lengths of closed cycle-free walks, i.e., there is only one closed cycle-free walk of length $2$ connected to one of the nodes of a $g$-cycle. This corresponds to a single category of CWWCs.
In the following, we partition this single category of $(g+2)$-CWWCs into two subcategories: {\bf 1.1} In this subcategory, the set of edges of the CWWC consists of the edges of a $g$-cycle and one extra edge that is incident to one node of the $g$-cycle. See Fig. \ref{graphA2}(a),  for an example. {\bf 1.2} In this subcategory, the set of edges of the CWWC consists of the edges of a $g$-cycle, i.e., the closed cycle-free walk of length $2$ is connected to one of the nodes of the cycle and traverses one of the edges of the cycle. See Fig. \ref{graphA2}(b), for an example.

To find a $(g+2)$-CWWC in Subcategory 1.1, first, we need to choose a cycle of length $g$ (the graph has $N_g$ cycles of length $g$). Then, we need to choose a node $v$ from the cycle (the cycle has $g/2$ variable nodes and $g/2$ check nodes). Finally, we need to choose an edge which is incident to $v$ and is not in the cycle (each variable node has $d_v-2$ such edges, and each check node has $d_c-2$ such edges).  Consequently, the total number of $(g+2)$-CWWCs in Subcategory 1.1 
is equal to $N_g \times [\frac{g}{2}(d_v-2)+\frac{g}{2}(d_c-2)] = N_g \times [\frac{g}{2}(d_v+d_c)-2g]$. 

In order to find a $(g+2)$-CWWC in Subcategory 1.2, we need to choose a cycle (i.e., $N_g$ options), and an edge from that cycle (i.e., $g$ options).  This amounts to $N_g \times g $ choices. We thus have $\Psi_{g+2}(d_v,d_c,G)=N_g\times [\frac{g}{2}(d_v+d_c)-g] \times 2(g+2)$. 
\end{proof}

\begin{figure}[ht]
\begin{center}
\includegraphics[scale=.45]{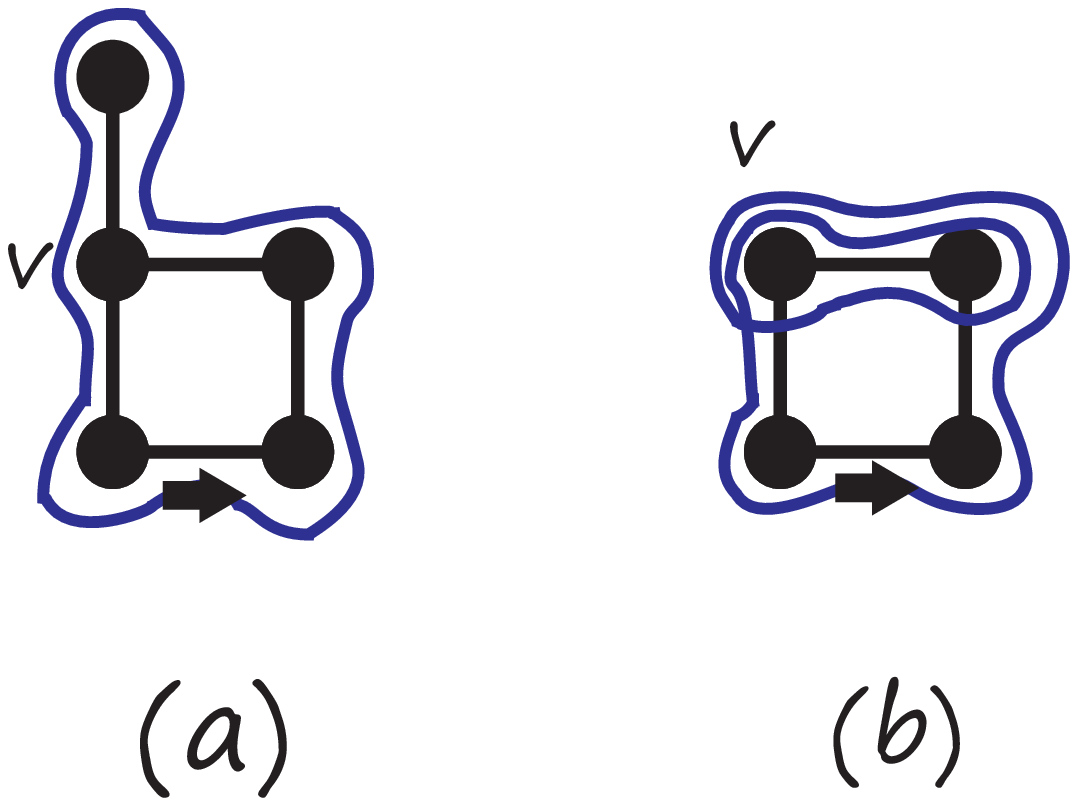}
\caption{(a) A $6$-CWWC in Subcategory 1.1, (b) A $6$-CWWC in Subcategory 1.2.
} \label{graphA2}
\end{center}
\end{figure}

\begin{ex}\label{Example2}
Consider the complete bipartite graph $K_{x,x}$. Using combinatorial arguments, in Example \ref{Example1}, we showed that for this graph, $N_6=x^2(x-1)^2(x-2)^2/6$.
Now, we use Theorems \ref{L1} and \ref{L2}, to calculate the number of $6$-cycles.
We have $\sum_j \lambda_j^6 = tr ({A(K_{x,x})^6})=2x^6$. From Table \ref{T1}, we obtain $S_{x,x,6}= x(x^2 + 2x(x-1) + 2(x-1)^2 )$.
Thus, by (\ref{E2}), we have $\Omega_{6}(x,x,G)=2x\times x(x^2 + 2x(x-1) + 2(x-1)^2 )$.
Also, $N_g\times [\frac{g}{2}(d_v+d_c)-g]\times 2(g+2)= \frac{x^2(x-1)^2}{4}(2(2x)-4)\times 12$.
Consequently, we have $N_6=x^2(x-1)^2(x-2)^2/6$.
\end{ex}

\subsection{Calculation of $N_{g+4}$}
In the following, we calculate $\Psi_{g+4}(d_v,d_c,G)$. This together with Theorem~\ref{L1} and (\ref{E2}) are then used to compute $N_{g+4}$.

\begin{theo}\label{L3}
For any $(d_v, d_c)$-regular bipartite graph $G$ with girth $g$ at least six, we have
\begin{align*}
\dfrac{\Psi_{g+4}(d_v,d_c,G)}{2(g+4)}&= N_{g+2}\times [\frac{g+2}{2}(d_v+d_c)-(g+2)]\\
                       &+ N_g \times [\frac{g}{2}(d_v-2)(d_c-1)+\frac{g}{2}(d_c-2)(d_v-1)]\\
                       &+ N_g \times \Big( [{{\frac{g}{2}} \choose 2}+\frac{g}{2}](d_v-2)^2+ [{{\frac{g}{2}} \choose 2}+\frac{g}{2}](d_c-2)^2 + (\frac{g}{2})^2(d_v-2)(d_c-2) \Big)\\
                       &+ N_g \times \Big({g\choose 2}+2g + (g+2)\times (\frac{g}{2}(d_v-2)+\frac{g}{2}(d_c-2))\Big)\:,
\end{align*}
where $N_g$ and $N_{g+2}$ are the number of cycles of length $g$ and $g+2$ in $G$, respectively.
\end{theo}

\begin{proof}
In the following, we prove that the right hand side of the above equation calculates the number of CWWCs of length $g+4$ not taking into account the starting edge or the direction of the walks.
Since the girth of the graph is at least six, by Lemma \ref{L5}, every $(g+4)$-CWWC consists of one cycle and some closed cycle-free walks. 
Depending on the length of the cycle, two cases are possible: \\
{\bf Case 1.} The closed walk consists of  a $(g+2)$-cycle and some closed cycle-free walks. In this case, similar to the case of Theorem \ref{L2}, there is one category of CWWCs with two subcategories: (i) The set of the edges of the closed walk consists of the edges of a $(g+2)$-cycle and one extra edge that is incident to one node of the $(g+2)$-cycle. (ii) The set of the edges of the closed walk is the same as the set of the edges of a $(g+2)$-cycle. Similar to Theorem \ref{L2}, the total number of $(g+4)$-CWWCs in this case is equal to $N_{g+2}\times [\frac{g+2}{2}(d_v+d_c)-(g+2)]$.\\
{\bf Case 2.} The closed walk consists of a $g$-cycle and some closed cycle-free walks. For this case, $p(2)=2$, and there are two categories of
$(g+4)$-CWWCs. In the first category, each CWWC consists of one $g$-cycle $\zeta$ and one closed cycle-free walk of length $4$ from one of the nodes of $\zeta$. In the second category, each CWWC consists of one $g$-cycle  $\zeta$ and two closed cycle-free walks, each of length $2$, from two nodes of $\zeta$. The CWWCs in the first and second categories can be partitioned into eight and three subcategories, respectively. See, Figures \ref{graphA4} 
and \ref{graphA3}, respectively. In the following, we count the number of CWWCs in each subcategory of each category by referring to the structure of corresponding CWWCs as shown in Figures \ref{graphA4} and \ref{graphA3}.

\begin{figure}[ht]
\begin{center}
\includegraphics[scale=.45]{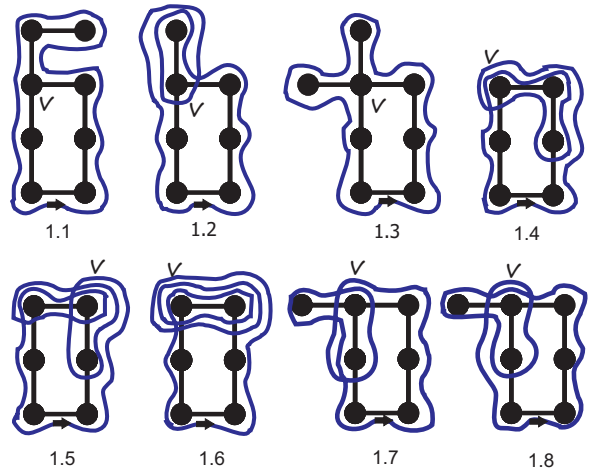}
\caption{Different subcategories of $(g+4)$-CWWCs in Category 1 of Case 2.
}
\label{graphA4}
\end{center}
\end{figure}

\begin{figure}[ht]
\begin{center}
\includegraphics[scale=.45]{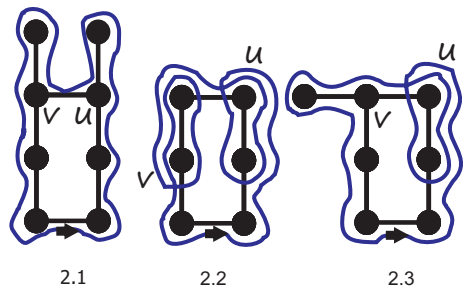}
\caption{Different subcategories of $(g+4)$-CWWCs in Category 2 of Case 2.
}
\label{graphA3}
\end{center}
\end{figure}

\underline{Category 1:}\\ 
{\bf 1.1} To find a $(g+4)$-CWWC in this subcategory, first, we need to choose a cycle of length $g$ (the graph has $N_{g}$ cycles of length $g$). Then we choose a node $v$ from the cycle (the cycle has $\frac{g}{2}$ variable nodes and $\frac{g}{2}$ check nodes). Next, we choose an edge which is incident to $v$ and is not in the cycle (each variable node has $d_v-2$ such edges and each check node has $d_c-2$ such edges). Finally, we choose an edge 
which shares a node with the edge that we picked in the previous step. Consequently, the total number of $(g+4)$-CWWCs in this subcategory is equal to $N_g \times [\frac{g}{2}(d_v-2)(d_c-1)+\frac{g}{2}(d_c-2)(d_v-1)]$.\\
{\bf 1.2} It is easy to see that the multiplicity of CWWCs in this subcategory is equal to  $N_{g}\times [\frac{g}{2}(d_v-2) +\frac{g}{2}(d_c-2) ]$.\\
{\bf 1.3} For this subcategory, the number of CWWCs is given by $$N_{g}\times \Big(\frac{g}{2}(d_v-2)(d_v-3) +\frac{g}{2}(d_c-2)(d_c-3)\Big).$$
{\bf 1.4} To find a CWWC in this subcategory, we need to choose a cycle of length $g$, first ($N_g$ possibilities), and then choose two adjacent edges from that cycle ($g$ possibilities). The number of CWWCs in this subcategory is thus equal to $N_g \times g$. \\
{\bf 1.5} Similar to the previous subcategory, for this one also the multiplicity of CWWCs is equal to  $N_g \times g$. \\
{\bf 1.6} The multiplicity for this subcategory is also given by $N_g \times g$.\\
{\bf 1.7} To find a CWWC in this subcategory, we need to first choose a $g$-cycle ($N_g $ possibilities). Then we need to choose a  node $v$ from that cycle and also an edge  connected to $v$ that is not part of the cycle 
($\frac{g}{2}(d_v-2)+\frac{g}{2}(d_c-2)$ possibilities). Finally, we need to choose one of the two edges which are incident to $v$ and are in the selected cycle ($2$ possibilities). The total number of CWWCs in this subcategory is thus 
$N_g \times  (\frac{g}{2}(d_v-2)+\frac{g}{2}(d_c-2))\times 2$.\\
{\bf 1.8} Similar to Subcategory 1.7, the number of CWWCs in this subcategory is $N_g \times  (\frac{g}{2}(d_v-2)+\frac{g}{2}(d_c-2))\times 2$.\\

\underline{Category 2:} \\
{\bf 2.1} To find a CWWC in this subcategory, first, we choose a cycle of length $g$. Then, we choose two nodes $v$ and $u$ from the cycle (these two nodes can be both variable nodes, both check nodes, or one variable node and 
one check node). Finally, for each selected node, we choose an edge which is incident to that node and is not in the cycle. Thus, the total number of CWWCs in this subcategory is equal to $$N_g \times \Big( {{g/2} \choose 2}(d_v-2)^2+ {{g/2} \choose 2}(d_c-2)^2 + (\frac{g}{2})^2(d_v-2)(d_c-2) \Big).$$
{\bf 2.2} For a CWWC in this subcategory, we choose a $g$-cycle and two edges from that cycle that are not incident. The multiplicity is thus $$N_g \times [{{g}\choose 2}-g].$$\\
{\bf 2.3} In this case, we first choose a $g$-cycle. Then, we choose a node $v$ and an edge from that cycle such that the selected edge is not incident to $v$ ($g-2$ possibilities). Finally, we choose an edge incident to $v$ which does not belong to the selected $g$-cycle.  The multiplicity in this case  is thus $N_g \times (g-2)\times (\frac{g}{2}(d_v-2)+\frac{g}{2}(d_c-2))$.

Adding up the multiplicities derived above, we obtain the total multiplicity of $(g+4)$-CWWCs given in the theorem.
\end{proof}

\subsection{Calculation of $N_{i}$ for $g+6 \leq i \leq 2g-2$}

To compute $N_i$, for $g+6 \leq i \leq 2g-2$, using Theorem~\ref{L1}, one needs to calculate the corresponding $\Psi_i$. Such a calculation involves steps similar to those taken in Theorems~\ref{L2} and \ref{L3}.
For each value of $i$ in the above range, based on Lemma~\ref{L5}, all CWWCs consist of a single cycle and some closed cycle-free walks from the nodes of that cycle. The CWWCs should then be 
first partitioned based on the length of the cycle ($g, g+2, \ldots, i-2$), and then for each cycle length $g+2k, k=0, \ldots, \frac{i-g}{2}-1$, they should be further partitioned into different 
categories based on the possible lengths of the closed cycle-free walks (the number of categories is equal to $p(\frac{i-g-k}{2})$). Within each category, subcategories then need to be identified based on different ways that the closed cycle-free walks can be attached to the cycle. 

It is easy to see that calculation of each $\Psi_i$ requires the information of all $N_j, g \leq j \leq i-2$. One can also see that the calculations required for finding the multiplicities of $i$-CWWCs that consist of a $j$-cycle are
similar to those required for finding the multiplicities of $(i-2)$-CWWCs that consist of a $(j-2)$-cycle. This can be seen, for example, by comparing the result of Theorem~\ref{L2} and 
the first term in the right hand side of the equation in Theorem~\ref{L3}. Similarly, the multiplicity of $(g+6)$-CWWCs that consist of a $(g+4)$-cycle is given by
$$
N_{g+4}\times [\frac{g+4}{2}(d_v+d_c)-(g+4)] \times 2(g+6)\:.
$$
Also, based on the calculations of Theorem~\ref{L3}, the number of $(g+6)$-CWWCs that consist of a $(g+2)$-cycle is
\begin{align*}
       N_{g+2}&\Big[ [\frac{g+2}{2}(d_v-2)(d_c-1)+\frac{g+2}{2}(d_c-2)(d_v-1)]\\
                      &+ \Big( [{{\frac{g+2}{2}} \choose 2}+\frac{g+2}{2}](d_v-2)^2+ [{{\frac{g+2}{2}} \choose 2}+\frac{g+2}{2}](d_c-2)^2 + (\frac{g+2}{2})^2(d_v-2)(d_c-2) \Big)\\
                      &+ \Big({g+2\choose 2}+2(g+2) + (g+4)\times (\frac{g+2}{2}(d_v-2)+\frac{g+2}{2}(d_c-2))\Big)\Big] \times 2(g+6)  .
\end{align*}

\section{Computing the number of $4$-cycles ($6$-cycles) in irregular (half-regular) bipartite graphs}
\label{sec35}

\subsection{Calculation of $N_{4}$ for irregular graphs with $g \geq 4$}
\label{subsec44}

\begin{theo}\label{Th001}
In an irregular bipartite graph $G$ with the node set $V(G)$, we have
\begin{equation}\label{EEE4}
N_4=\frac{\sum_{j=1}^{|V(G)|} \lambda_j^4 - \sum_{v \in V(G)} d(v) \big(2d(v)-1\big)}{8}\:,
\end{equation}
where $N_4$ and  $\{\lambda_j\}$ are the number of $4$-cycles and the spectrum of $G$, respectively.
\end{theo}

\begin{proof}
Let $V(G)=U \cup W$, where $U=\{u_1, u_2, \ldots , u_n\}$, $W=\{w_1,w_2, \ldots, w_{m}\}$, and in which, the degree of node $u_i$ is $d_i$ and the degree of node $w_i$ is $d_i'$.
In $G$, the set of closed walks of length $4$ can be partitioned into two categories: (1) $4$-cycles (2) closed cycle-free walks of length $4$. Let $\Omega_{i}(G), i \geq 2$, denote the number of closed cycle-free walks of length $i$ in $G$, 
and $S_{u,G,i}$ ($S_{w,G,i}$) be the number of closed cycle-free walks of length $i$ from the variable node $u$ (the check node $w$) to itself in $G$. 
We have
\begin{equation}
\Omega_{i}(G)=\sum_{u\in U}S_{u,G,i} + \sum_{w\in W}S_{w,G,i}\:,
\label{eq21}
\end{equation}
and
\begin{equation}
N_4 = \frac{\sum_{j=1}^{|V(G)|} \lambda_j^4 - \Omega_{4}(G)}{8}\:.
\label{eq22}
\end{equation}
We can have three different types of closed cycle-free walks of length $4$. See Fig. \ref{graphGG1}.
The number of closed cycle-free walks  of length $4$ of Type 1 from $u_j$ to itself  is $d_{j}(d_j-1) $. That number for Type 2 is $d_j$, and for Type 3 is $ \sum_{w_k \in N(u_j)}(d_k'-1)$,
where $N(u_j)$ is the set of neighbors of $u_j$. Thus, 
\begin{equation}\label{EEE1}
S_{u_j,G,4}=d_{j}^2 + \sum_{w_k \in N(u_j)}( d_k'-1) \:.
\end{equation}
Similarly, for a check node $w_j$, we have:
\begin{equation}\label{EEE2}
S_{w_j,G,4}=(d_{j}')^2 + \sum_{u_k \in N(w_j)}( d_k-1)\:.
\end{equation}
By (\ref{EEE1}):
\begin{align*}
\sum_{u_j\in U} S_{u_j,G,4}   &= \sum_{u_j\in U} d_{j}^2 + \sum_{u_j\in U}\sum_{w_k \in N(u_j)}( d_k'-1) \\
                              &= \sum_{u_j\in U}  d_{j}^2+ \sum_{w_k \in W}d_k'( d_k'-1). \numberthis \label{EEE9}
\end{align*}
Similar to (\ref{EEE9}), we have:
\begin{equation}\label{EEE10}
\sum_{w_j\in W} S_{w_j,G,4} = \sum_{w_j\in W} ( d_{j}')^2+ \sum_{u_k \in U} d_k( d_k-1)\:.
\end{equation}
By using (\ref{EEE9}) and (\ref{EEE10}) in (\ref{eq21}), we have:
\begin{equation}\label{EEE3}
\Omega_{4}(G)= \sum_{u_j\in U} d_j(2d_j-1)+ \sum_{w_j\in W} d_j'(2d_j'-1)\:.
\end{equation}
Combining (\ref{EEE3}) with (\ref{eq22}) completes the proof.
\end{proof}

\begin{figure}[ht]
\begin{center}
\includegraphics[scale=.4]{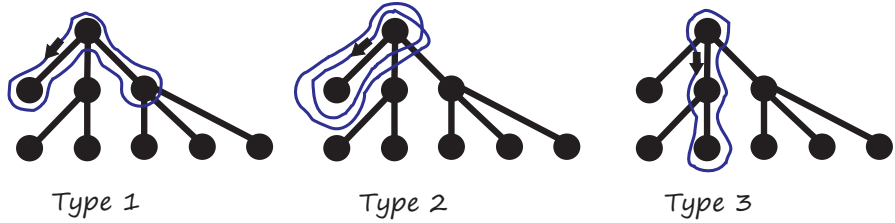}
\caption{The three different types of closed cycle-free walks of length $4$.
} \label{graphGG1}
\end{center}
\end{figure}

It can be seen that the result of Theorem~\ref{Th001} reduces to that of Theorem \ref{L1} for $4$-cycles in bi-regular bipartite graphs.


\begin{ex}
Consider a path $P_5$ of length $4$. The spectrum of $P_5$ is $\{\sqrt{3},1,0,-1,-\sqrt{3}\}$.
Thus, we have $\sum \lambda_j^4= 20$. For $P_5$, we also have $\sum_{u_j\in U} d_j(2d_j-1)+ \sum_{w_j\in W} d_j'(2d_j'-1)= 20$. 
Thus, by Theorem \ref{Th001}, $N_4 = (20-20)/ 8 =0$, which is clearly the correct answer.
\end{ex}



\subsection{Calculation of $N_{6}$ for half-regular bipartite graphs with $g \geq 6$}

The positive result of Subsection~\ref{subsec44} is applicable to half-regular graphs and can be used to compute $N_4$.  In this subsection, we compute $N_6$ for half-regular bipartite graphs with $g \geq 6$, in terms of graph's spectrum and degree sequences. In the following, without loss of generality, we assume that the graphs are regular on the variable side, i.e., they are variable-regular.

\begin{theo}\label{Th00v2}
Let $G$ be a variable-regular bipartite graph with girth at least six and node set $V(G)=U \cup W$, where $U=\{u_1, u_2, \ldots , u_n\}$, $W=\{w_1,w_2, \ldots, w_{m}\}$, and
in which, the degree of every node $u_i \in U$ is $d_v$, and the degree of node $w_i$ is $d_i'$. We then have
\begin{align*}
12N_6  &= \sum_{j=1}^{|V(G)|} \lambda_j^6 - n \times d_v\Big(1+3(d_v-1)+2(d_v-1)(d_v-2) \Big)\\
       &- \sum_{w_j\in W}  \Big( d_j'(3d_j'-2)+2d_j'(d_j'-1)(d_j'-2)+6d_j'(d_j'-1)(d_v-1)\Big)\\
       &- \sum_{w_j\in W}  \Big(  3d_j'(d_j'-1+d_v-1) \Big)\:,
\end{align*}
where $N_6$ and $\{\lambda_j\}$ are the number of $6$-cycles and the spectrum of $G$, respectively.
\end{theo}

\begin{proof}
For the graph $G$, we partition the set of closed walks of length $6$ into two categories: (1) $6$-cycles (2) Closed cycle-free walks of length $6$. 
Let $\Omega_{6}(G)$ denote the number of closed cycle-free walks of length $6$ in the graph $G$, and let $S_{u,G,6}$ ($S_{w,G,6}$) be the number of closed cycle-free walks of length $6$ from the variable node
$u$ (check node $w$) to itself. We thus have: $\Omega_{6}(G)=\sum_{u\in U}S_{u,G,6} + \sum_{w\in W}S_{w,G,6}$.
We have twelve different types of closed cycle-free walks of length $6$. See Fig. \ref{graphGGFF}.

\begin{figure}[ht]
\begin{center}
\includegraphics[scale=.4]{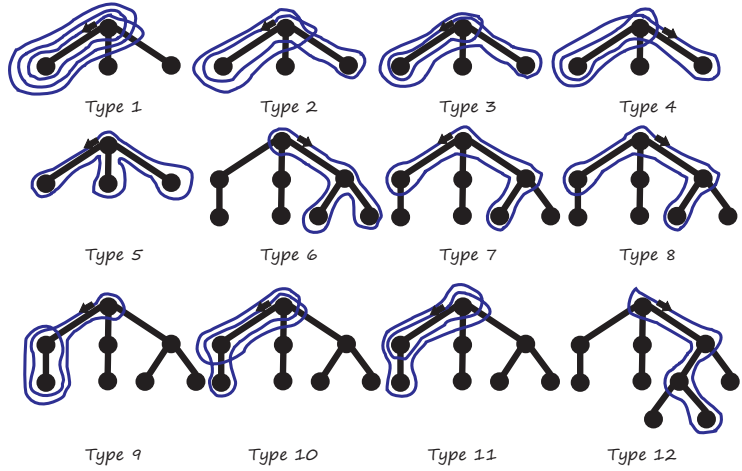}
\caption{The twelve different types of closed cycle-free walks of length $6$.
} \label{graphGGFF}
\end{center}
\end{figure}

Next, we calculate the number of closed cycle-free walks of length $6$ for each type:\\
{\bf Type 1} The number of closed walks of Type 1 from any variable node $u_j$ to itself is $d_v$. This number for the check node $w_j$ is $d_j'$. Thus, the total number is
\begin{equation}\label{ES1}
n \times d_v+\sum_{w_j\in W} d_j'\:.
\end{equation}
{\bf Types 2, 3, 4} The total number of closed walks of Types 2, 3, and 4 from any variable node $u_j$ to itself is $3d_v(d_v-1)$. Similarly, for the check node $w_j$, this number is  $3d_j'(d_j'-1) $.
The total number is thus
\begin{equation}\label{ES2}
n \times 3d_v(d_v-1)+\sum_{w_j\in W} 3d_j'(d_j'-1)\:.
\end{equation}
{\bf Type 5} The number of closed walks of Type 5 from any variable node $u_j$ to itself  is $d_v(d_v-1)(d_v-2) $.
Similarly, for the check node $w_j$, the number  is $d_j'(d_j'-1)(d_j'-2) $. So, the total number is
\begin{equation}\label{ES3}
n \times d_v(d_v-1)(d_v-2)+\sum_{w_j\in W} d_j'(d_j'-1)(d_j'-2)\:.
\end{equation}
{\bf Type 6} The number of closed walks  of Type 6 from the variable node $u_j$ to itself  is
$\sum_{w_k \in N(u_j)} ( d_k'-1)( d_k'-2).$ Similarly, for the check node $w_j$, the number is
$\sum_{u_k \in N(w_j)} ( d_v-1)( d_v-2),$ and thus the total number is
\begin{equation}\label{ES4}
n \times d_v(d_v-1)(d_v-2)+\sum_{w_j\in W} d_j'(d_j'-1)(d_j'-2)\:.
\end{equation}
{\bf Types 7, 8} The total number of closed walks of Types 7 and 8 from the variable node $u_j$ to itself  is
$2(d_v-1)\sum_{w_k \in N(u_j)} ( d_k'-1).$ Similarly, for the check node $w_j$, the number is $2d_j'(d_j'-1)(d_v-1)$.
Thus, the total number is
\begin{equation}\label{ES5}
4(d_v-1)\sum_{w_j\in W}  d_j'(d_j'-1)\:.
\end{equation}
{\bf Types 9, 10, 11} The total number of closed walks of Types 9, 10 and 11 from the variable node $u_j$ to itself  is
$3\sum_{w_k \in N(u_j)} ( d_k'-1).$ Similarly, for the check node $w_j$, the number  is $3d_j' (d_v-1)$.
Thus the total number is
\begin{equation}\label{ES6}
\sum_{w_j\in W}  \Big( 3d_j'(d_j'-1)+ 3d_j'(d_v-1) \Big)\:.
\end{equation}
{\bf Type 12} The number of closed walks  of Types 12 from the variable node $u_j$ to itself  is
$\sum_{w_k \in N(u_j)} \sum_{u_\ell \in N(w_k)}( d_v-1)$.
Similarly, for the check node $w_j$, the number is $\sum_{u_k \in N(w_j)} \sum_{w_\ell \in N(u_k)}( d_{\ell}'-1) $.
Thus the total number is
\begin{equation} \label{W1E2}
\sum_{u_j\in U} \sum_{w_k \in N(u_j)} \sum_{u_\ell \in N(w_k)}(d_v-1) + \sum_{w_j\in W}\sum_{u_k \in N(w_j)} \sum_{w_\ell \in N(u_k)}( d_{\ell}'-1)\:.
\end{equation}
The two terms in (\ref{W1E2}) can be simplified as follows:
$$ \sum_{u_j\in U} \sum_{w_k \in N(u_j)} \sum_{u_\ell \in N(w_k)}( d_v-1) = \sum_{w_j\in W}    d_j'(d_j'-1)(d_v-1)\:,$$
and
$$  \sum_{w_j\in W}\sum_{u_k \in N(w_j)} \sum_{w_\ell \in N(u_k)}( d_{\ell}'-1) = \sum_{w_j\in W}    d_j'(d_v-1)(d_j'-1)\:.$$
Consequently, the total number in (\ref{W1E2}) can be written as
\begin{equation}\label{ES7}
2 (d_v-1) \sum_{w_j\in W}   d_j'(d_j'-1)\:.
\end{equation}

By adding up (\ref{ES1}), (\ref{ES2}), (\ref{ES3}), (\ref{ES4}), (\ref{ES5}), (\ref{ES6}) and (\ref{ES7}), we have
\begin{align*}
\Omega_{6}(G) &= n\times d_v\Big(1+3(d_v-1)+2(d_v-1)(d_v-2) \Big)\\
              &+ \sum_{w_j\in W}  \Big(d_j'(3d_j'-2)+2d_j'(d_j'-1)(d_j'-2)+6d_j'(d_j'-1)(d_v-1)\Big)\\
              &+ \sum_{w_j\in W}  \Big(   3d_j'(d_j'-1+d_v-1) \Big)
\end{align*}
This, together with $N_6=(\sum_j \lambda_j^6 - \Omega_{6}(G))/12$, complete the proof.
\end{proof}

One can see that the result of Theorem~\ref{Th00v2} reduces to that of Theorem \ref{L1}, for the special case of bi-regular bipartite graphs with $g \geq 6$.

\begin{figure}[ht]
\begin{center}
\includegraphics[scale=.3]{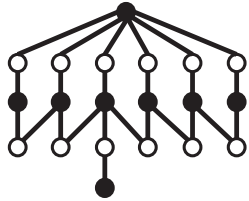}
\caption{The variable-regular graph $G$ of Example~\ref{ex-ty}.
} \label{graphGG9}
\end{center}
\end{figure}

\begin{ex}
\label{ex-ty}
Consider the variable-regular bipartite graph $G$ shown in Fig. \ref{graphGG9}. In $G$, we have $d_v=2$, and the degree sequence of check nodes is $(6,4,3,3,3,2,2,1)$. 
Also, $n=12$ and $m=8$. By inspection, it is clear that $G$ has five $6$-cycles. Now, we compute the number of $6$-cycles by Theorem \ref{Th00v2}. 
We have $tr({A(G)^6})=1344$, and $\Omega_{6}(G)=1284$. Thus,  $N_6 = (1344-1284)/ 12 = 5$.
\end{ex}

\section{Numerical results}
\label{sec4}

In this section, we compute the multiplicity of short cycles in the Tanner graphs of two well-known LDPC codes using the closed form formulas that we derived in previous sections.
We then compare the results with those obtained by the backtracking algorithm of \cite{Jeff} to verify that they match.

\subsection{Tanner $(155,64)$ code}
As the first example, we consider the Tanner $(155,64)$ code~\cite{mackay2005encyclopedia}.
This code is a $(3, 5)$-regular LDPC code with $n=155$, $m=93$ and $g=8$.
To compute $N_8$, we calculate $tr ({A(G)^{8}})=475230$. From Table \ref{T1}, we obtain $S_{3,5,8}=1509$ and $S_{5,3,8}=2515$. Thus, by (\ref{E2}), $\Omega_{8}(3,5, G)=  467790$, and using Theorem \ref{L1},  we 
compute $N_8=(475230-467790)/16= 465$. 
For $N_{10}$, we have  $tr ({A(G)^{10}})= 4636050$, and from Table \ref{T1}, we obtain $S_{3,5,10}=13995$ and $S_{5,3,10}=23325$.
Thus, by (\ref{E2}), we have $\Omega_{10}(3,5, G)= 4338450$.
Also, $N_g\times [\frac{g}{2}(d_v+d_c)-g]\times 2(g+2)= 223200$. Consequently,  by Theorem \ref{L1}, we have $N_{10}=(4636050-4338450-223200)/(2 \times 10)= 3720$. 

Finally, we have $tr ({A(G)^{12}})= 49222110$. From Table \ref{T1}, $S_{3,5,12}= 134277$ and $S_{5,3,12}=223795$. Thus, by (\ref{E2}), we have $\Omega_{12}(3,5, G) = 41625870$.
Also, by Theorem \ref{L3}, $\Psi_{12}(3,5,G)=  7053120$. Consequently, by Theorem \ref{L1}, we have $N_{12}=(49222110-41625870-7053120)/24= 22630$.

\subsection{Margulis $(2640, 1320)$ code}

As the second example, we consider Margulis $(2640, 1320)$ code with $g = 8$~\cite{mackay2005encyclopedia}.  
This code is a $(3, 6)$-regular LDPC code with $n=2640$ and $m=1320$. 
For Tanner graph $G$ of this code, we have $tr ({A(G)^{8}})=11774400$. 
From Table \ref{T1}, $S_{3,6,8}= 2226$ and $S_{6,3,8}= 4452$. Thus, by (\ref{E2}), $\Omega_{8}(3,6, G)=  11753280$. 
By Theorem \ref{L1}, we then have $N_8=(11774400-11753280)/16= 1320$. 
To compute $N_{10}$, we first obtain $tr ({A(G)^{10}})= 124924800$. From Table \ref{T1}, $S_{3,6,10}= 23478$ and $S_{6,3,10}= 46956$.
Thus, by (\ref{E2}), $\Omega_{10}(3,6, G)=  123963840$. Also, $N_g\times [\frac{g}{2}(d_v+d_c)-g]\times 2(g+2)=  739200$.
By Theorem \ref{L1}, we thus have $N_{10}=(124924800-123963840-739200)/(2 \times 10)= 11088$. 

Finally for computing $N_{12}$, $tr ({A(G)^{12}})= 1382325120$, and from Table \ref{T1}, $S_{3,6,12} =256374$, $S_{6,3,12} =512748$. 
Thus, by (\ref{E2}), we have $\Omega_{12}(3,6, G)= 1353654720$. Also, by Theorem \ref{L3}, $\Psi_{g+4}(d_v,d_c,G)= 26104320$.
Consequently, by Theorem \ref{L1}, we compute $N_{12}=(1382325120-1353654720-26104320)/24=106920$.

All the above results for $N_8$, $N_{10}$ and $N_{12}$ match those from \cite{Jeff}.

\section{Concluding remarks}
\label{sec5}

It has been long known that the number of closed walks in a graph can be computed using the spectrum of the graph. Very recently, Blake and Lin~\cite{blake2017short} computed the number of shortest cycles in a bi-regular bipartite graph in terms of the spectrum of the graph and the extra information of node degrees and multiplicities on the two sides of the bipartition. In this work, we extended the results of~\cite{blake2017short} to compute the multiplicity of $i$-cycles for $g+2 \leq i \leq 2g-2$, in bi-regular bipartite graphs, as a function of the spectrum and the node degrees. 
Moreover, for irregular  (half-regular) bipartite graphs with $g \geq 4$  ($g \geq 6$), we derived closed form equations for the multiplicity of $4$-cycles ($6$-cycles) in terms of the spectrum and the degree sequences. 
 
 In the context of coding, the degree sequences of Tanner graphs play an important role in the performance of the corresponding LDPC codes, particularly in the waterfall region. Given a degree distribution, however,
 it is well-known that the performance of practical finite-length codes can have a large variation in the error-floor region. The main cause of such a large variation is the difference in the trapping set distribution of different codes (all with the same degree distribution). Trapping sets, on the other hand, are closely related to the distribution of short cycles in the graph. The results of this work show that, for a given degree distribution, it is in fact the spectrum of the Tanner graph that is responsible for the variations in cycle multiplicities. In this context, it would be interesting to study the relationship between the spectrum of the graph and its trapping set distribution.

\bibliographystyle{ieeetr}

\end{document}